\documentclass[leqno, 11pt]{amsart}

\usepackage{enumerate}

\newcommand\bq{{\mathbf q}}

\newcommand\be{{\mathbf e}}

\newcommand\bpi{{\boldsymbol{\pi}}}

\newcommand\LL{{\mathbb L}}
\newcommand\EE{{\mathbb E}}
\newcommand\PP{{\mathbb P}}

\newcommand\RR{{\mathbb R}}
\newcommand\TT{{\mathbb T}}
\newcommand\ZZ{{\mathbb Z}}

\usepackage{amssymb, amsmath}

\newtheorem{theo}{Theorem}
\newtheorem{lemma}{Lemma}
\newtheorem{remark}{Remark}
\newtheorem{cor}{Corollary}

\title[]{Thermal conductivity for a noisy disordered harmonic chain}

\author{C\'edric Bernardin}
\address{%
Universit\'e de Lyon and CNRS, UMPA, UMR-CNRS 5669, ENS-Lyon,
46, all\'ee d'Italie, 69364 Lyon Cedex 07 - France.
}%

\date{\today}

\thanks{\textsc{Acknowledgements.} We thank A. Dhar, J.L. Lebowitz and S. Olla for their interest in this work. We acknowledge the support of the French Ministry of Education through the ANR BLAN07-2184264 grant.}


\begin{document}


\maketitle

\begin{abstract}
We consider a $d$-dimensional disordered harmonic chain (DHC) perturbed by an energy conservative noise. We obtain uniform in the volume upper and lower bounds for the thermal conductivity defined through the Green-Kubo formula. These bounds indicate a positive finite conductivity. We prove also that the infinite volume homogenized Green-Kubo formula converges. 
\end{abstract}

 \section{Introduction}
 \label{sec:introduction}
 
Thermal transport in Fermi-Pasta-Ulam (FPU) chains is a subject of intense research (\cite{BLL}, \cite{LLP}). In a perfect crystal, equilibrium positions of atoms form a perfect regular configuration (e.g. a sublattice of ${\mathbb Z}^d$). Due to interactions between nearest-neighbor atoms and with a given substrate, real position of atoms are subject to fluctuations around equilibrium. Lattice vibrations are the carriers of heat current.  FPU models are described by Hamiltonian of the form:
\begin{equation*}
{\mathcal H} = \sum_x \cfrac{p_x^2}{2m_x} + \sum_x W(q_x)+ \sum_{|x-y|=1} V(q_x-q_y)
\end{equation*}
Here $q_x \in {\mathbb R}^d$ is the deviation of atom $x$ from its equilibrium position, $p_x$ is its momentum and $m_x$ its mass. Interactions between atoms are described by the potential $V$ while the pinning potential $W$ is for the interaction with the substrate. The main problem is to understand the dependance of the thermal conductivity $\kappa_N$ with the size $N$ of the system. Fourier's law requires a finite positive limit of $\kappa_N$ in the thermodynamic limit $N \to \infty$. In low dimensional systems it is largely accepted that anomalous heat conduction takes place as soon as momentum is conserved ($W=0$). For nonlinear systems only few rigorous results exist and require extra assumptions (\cite{BK}) or start from kinetic approximations (\cite{ALS}, \cite{LS}). Moreover numerical simulations are not very conclusive (\cite{BDLLOP} and references therein). On the other hand homogenous harmonic chain is an exactly solvable model easy to study. Nevertheless it does not reproduce expected behavior of real systems and it turns out that  $\kappa_N$ is of order $N$. This is due to the fact that phonons can travel ballistically along the chain without scattering (\cite{RLL}). 

Recently it has been proposed to perturb homogenous FPU chains by a conservative stochastic noise. The perturbation of an harmonic chain is sufficient to reproduce qualitatively what is expected for real anharmonic chains (\cite{BBO1}, \cite{B}, \cite{BO}). The dynamics becomes non-linear because of the noise but remains mathematically reachable. Perturbation of anharmonic chains is more difficult to study and only partial but rigorous results have been obtained (\cite{BBO2}).

We turn now to non homogenous chains. As it is well known, the presence of disorder generally induces localization of  the normal modes  and one can expect the latter to behave as perfect thermal insulators ($\kappa_N \to 0$). The only analytically tractable model is the one dimensional disordered harmonic chain (DHC).  Surprisingly the behavior of the thermal conductivity depends on boundary conditions and on the properties of the thermostats. Consider first the unpinned DHC. For fixed boundary conditions (the Casher-Lebowitz model, \cite{CL}), $\kappa_N \sim N^{1/2}$, while for free boundaries (the Rubin-Greer model, \cite{RG}), $\kappa_N \sim N^{-1/2}$. This curious phenomenon has been studied in \cite{D} in a more general setting and it turns out that "the exponent [of $\kappa_N$] depends not only on the properties of the disordered chain itself, but also on the spectral properties of the heat baths. For special choices of baths one gets the "Fourier  behavior" ".  If we add a pinning potential in the DHC, $\kappa_N$ becomes exponentially small in $N$.  

Recently, Dhar and Lebowitz (\cite{DL}) were interested in the effect of both disorder and anharmonicity. The conclusions of their numerical simulations are that the introduction of a small amount of phonon-phonon interactions in the DHC leads to a positive finite thermal conductivity. Moreover it seems that the transition takes place instantaneously without any finite critical value of anharmonicity.  

In this paper we propose to study this question for the conservative perturbed model introduced in \cite{B}, \cite{BO}. Our results are valid in any dimension (DHC has only been studied in the one dimensional case). We consider DHC perturbed by a stochastic noise conserving energy and destroying all other conservation laws. In view of the numerical simulations of \cite{DL} one would expect the model to become a normal conductor : $\kappa_N \to \kappa$ with $\kappa$ finite and positive. The behavior of the thermal conductivity is here studied in the linear response theory framework by using the Green-Kubo formula. Curiously behavior of the conductivity defined through Green-Kubo formula has not been studied for DHC. It would be interesting to know what is the order of divergence of the latter. For the perturbed DHC we obtain uniform finite positive lower and upper bounds for the $d$ dimensional finite volume Green-Kubo formula of the thermal conductivity with or without pinning (Theorem \ref{th:QGK}) so that the thermal conductivity is always finite and positive. In particular it shows the presence of the noise is sufficient to destroy localization of eigen-functions in pinned DHC. Linear response approach avoids the difficulty to deal with a non-equilibrium setting where effects of spectral properties of heat baths could add difficulties as in the case of purely DHC. In the non-equilibrium setting, we expect that since the Green-Kubo formula for the thermal conductivity of the perturbed DHC remains finite, it will not depend on the boundaries. As a second result (Theorem \ref{th:AGK}) we show that the homogenized infinite volume Green-Kubo formula $\kappa_{hom.}$ is well defined, positive and finite. 

The paper is organized as follows. In the first section we define the dynamics. In section $3$, we give heuristic arguments predicting the equality between the Green-Kubo formula and the homogenized infinite volume Green-Kubo formula. The latter is proved to exist and to be finite and positive. In section $4$ we obtain uniform (in the volume) lower and upper bounds for the finite volume Green-Kubo formula.      
   
\vspace{0,5 cm}
   
{\textbf{Notations :}}  The canonical basis of $\RR^d$ is noted
$(e_1,e_2, \ldots, e_d)$ and the coordinates of a vector $u \in
\RR^d$ are noted $(u^1,\ldots,u^d)$. Its Euclidian norm $|u|$  is
given by $|u|=\sqrt{(u^1)^2+ \ldots +(u^d)^2}$ and the scalar
product of $u$ and $v$ is $u \cdot v$.  

If $N$ is a positive integer, $\mathbb T_N^d$ denotes the $d$-dimensional 
discrete torus of length $N$. We identify $\mathbb T_N^d =
\left(\mathbb Z /N{\mathbb Z}\right)^d$, i.e. $x = x + k N e_j$
for any $j= 1,\dots,d$ and $k\in \mathbb Z$.  

If $F$ is a function from ${\mathbb Z}^d$ (or $\TT_N^d$) into $\mathbb
R$ then the (discrete) gradient of $F$ in the direction $e_j$ is
defined by $(\nabla_{e_j} F)(x)= F(x +e_j) -F(x)$ and the
Laplacian of $F$ is given by $(\Delta F) (x)= \sum_{j=1}^d
\left\{F(x +e_j)+F(x-e_j)-2F(x)\right\}$.

\section{The dynamics of the closed system}
\label{sec:dynamics}

The Hamiltonian of a non homogenous harmonic chain of length $N$ with periodic boundary conditions is given by
\begin{equation*}
{\mathcal H} = \sum_{x \in {\mathbb T}_N^d} \cfrac{|p_x|^2}{2m_x} + \cfrac{1}{2}\sum_{x \in {\TT}_N^d} \left\{q_x \cdot (\nu I -\omega \Delta)q_x\right\}
\end{equation*} 
where $p_x=m_x v_x$, $v_x \in {\mathbb R}^d$ is the velocity of the particle $x$ and $m_x >0$ its mass. $q_x \in {\mathbb R}^d$ is the displacement of the atom $x$ with respect to its equilibrium position. Parameters $\omega$ and $\nu$ regulate the strength of the interaction potential $V(r)=\omega r^2$ and the strength of the pinning potential $W(q)=\nu q^2$. We perturb the harmonic chain by a conservative noise acting only on the velocities such it conserves the total kinetic energy $\sum_x {p_x^2}/{(2 m_x)}$. We define $\pi_x =m_x^{1/2} v_x$ and the generator of the noise is given by
  \begin{equation*}
S=\cfrac{1}{2d^2} \sum_{i,j,k=1}^d\sum_{x \in {\TT}_N^d } \left[Y^{i,j}_{x,x+e_k}\right]^2
\end{equation*}
with
\begin{equation*}
Y^{i,j}_{x,x+e_k}= \pi_{x+e_k}^j \partial_{\pi^i_{x}} - {\pi^i_{x}}\partial_{\pi^j_{x+e_k}}
\end{equation*}

We consider the stochastic dynamics corresponding to the generator
\begin{equation}\label{FP}
 L = A + \gamma S
\end{equation}
where $A$ is the usual Hamiltonian vector field
\begin{equation*}
A=\sum_{x \in \TT_N^d} \left\{ \cfrac{\partial H}{\partial{p_x}} \cdot \partial_{q_x} -\cfrac{\partial H}{\partial{q_x}} \cdot\partial_{p_x} \right\} 
\end{equation*}
The parameter $\gamma > 0$ regulates the strength of the noise. In terms of the $\pi$'s, the Hamiltonian vector field $A$ is given by
\begin{equation*}
A= \sum_{x \in \TT_N^d} \left\{ \cfrac{1}{\sqrt{m_x}} \pi_x \cdot \partial_{q_x} + \cfrac{1}{\sqrt{m_x}} \left[ \omega \Delta q_x -\nu q_x \right] \cdot \partial_{\pi_x}\right\}
\end{equation*}
and the Hamiltonian by
\begin{equation*}
H=\sum_{x \in \TT_N^d} \left\{ \cfrac{|\pi_x|^2}{2} +q_x (\nu I -\omega \Delta)q_x \right\}
\end{equation*}

Energy of atom $x$ is given by 
\begin{equation*}
e_x = \cfrac{|p_x|^2}{2 m_x} + \cfrac{\nu |q_x|^2}{2} + \cfrac{\omega}{2} \sum_{y; |y-x|=1} | q_y -q_x |^2 
\end{equation*}

The energy conservation law can be read locally as
\begin{equation*}
   e_x(t) -   e_x(0) =  \sum_{k=1}^d \left(
J_{x -e _k, x }(t) - J_{x ,x  +e _k}(t)\right)=-\sum_{k=1}^d \left( \nabla_{e_k} J_{x-e_k, x} \right) (t)
 \end{equation*}
where $J_{x ,x  +e _k}(t)$ is the total energy current
between $x $ and $x  +e _k$ up to
time $t$. This can be written as
\begin{equation}
  \label{eq:tc}
  J_{x , x  +e _k}(t)=\int_0^t j_{x , x  +e _k}(s) \; ds +
  M_{x , x  +e _k}(t) 
\end{equation}
In the above $M_{x , x  +e _k}(t)$ are  martingales and the instantaneous current $j_{x,x+e_k}$ is given by
\begin{equation*}
j_{x,x+e_k} = j^a_{x,x+e_k} + \gamma j^{s}_{x,x+e_k}
\end{equation*}
where $j^a_{x,x+e_k}$ is the Hamiltonian contribution
\begin{equation*}
j^{a}_{x,x+e_k} = -\cfrac{1}{\sqrt m_x} \pi_x \cdot (q_{x+e_k} -q_{x -e_k})
\end{equation*}
and $j^s_{x,x+e_k}$ is the noise contribution
\begin{equation*}
j^s_{x,x+e_k} = -\cfrac{1}{d} \nabla_{e_k} (|\pi_x|^2)
\end{equation*}

We consider the closed dynamics with periodic boundary conditions starting from the canonical Gibbs measure with temperature $T=\beta^{-1}$
\begin{equation*}
\mu^N_{\beta} (d\bpi,d\bq) = Z_{N,\beta}^{-1} \exp\left( -\beta H \right) d\bpi d\bq
\end{equation*}
The law of the process starting from $\mu^N_\beta$ is noted ${\mathbb P}_\beta$.

The conductivity in the direction $e_1$ is defined by the Green-Kubo
formula as the limit (when it exists)
\begin{equation}
\label{eq:gk0}
\kappa^{1,1}(\{m_x\}) = \lim_{t\to\infty}  \lim_{N\to\infty} \frac 1{2 T^2 t} \frac 1{N^{d}}{\mathbb E}_\beta
\left(\left[\sum_{x \in
     \mathbb T_N^d} J_{x ,x +e _1}(t)\right]^2 \right) 
\end{equation}
Because of the periodic boundary conditions, since $j^s$ if a
\emph{gradient}, the corresponding  terms cancel,
and we can write 
\begin{equation}
\label{decomp0}
\begin{array}{lcl}
     \sum_x  J_{x ,x +e _1}(t) &=& \int_0^t \sum_x  j^a_{x ,
       x +e _1}(s) \; ds + \sum_x  M_{x ,x +e _1}(t)\\
       &=& \int_0^t {\frak J}_{e _1}(s) \; ds + {\frak M}_{e _1}(t)
\end{array}
\end{equation}
so that 
\begin{equation}
\label{eq:2005}
\begin{split}
&(t N^d)^{-1}{\mathbb E}_\beta\left(\left[\sum_x  J_{x ,x +\be_1}(t)\right]^2 \right)\\
&= (tN^d)^{-1}{\mathbb E}_\beta \left(\left[\int_0^t {\frak J}_{e _1}(s)
    ds\right]^2 \right)+(tN^{d})^{-1}{\mathbb E}_\beta \left({\frak M}^2_{e _1}(t)\right)\\ 
&\qquad +2(tN^d)^{-1}{\mathbb E}_\beta \left(\left[\int_0^t {\frak J}_{e _1}(s) \; ds\right]{\frak  M}_{e _1} (t) \right) 
\end{split}
\end{equation}

The third term on the RHS of (\ref{eq:2005}) is zero by a
time reversal argument and the martingale term gives a $\gamma/d$ contribution (see \cite{BBO2} for a proof).
\begin{equation}
  \label{eq:decorre}
  \begin{split}
    \frac 1{2 T^2 t} \frac 1{N^{d}}{\mathbb E}_\beta \left(\left[\sum_x 
        J_{x ,x +e _1}(t)\right]^2 \right)& \\
    = (2 T^2 t N^d)^{-1}&{\mathbb E}_\beta \left(\left[\int_0^t {\frak J}_{e_1} (s) ds\right]^2 \right) + \frac{\gamma}d 
\end{split}
\end{equation}


In order to study the large time behavior of 
\begin{equation*}
C (t)= \lim_{N \to \infty}  (2 T^2 t N^d)^{-1}{\mathbb E}_\beta \left(\left[\int_0^t {\frak J}_{e_1} (s) ds\right]^2 \right)
\end{equation*}
we study the asymptotics as $\lambda \to 0$ of the Laplace transform ${\frak L} (\lambda)$ of $ t C(t)$
\begin{equation*}
{\frak L} (\lambda)= \int_{0}^\infty e^{-\lambda t} t C(t) dt
\end{equation*}
By stationarity and integration by parts, we have
\begin{equation*}
{\frak L}(\lambda)=\lim_{N \to \infty} \cfrac{1}{\lambda^2 T^2} \int_0^\infty dt e^{-\lambda t} {\mathbb E}_\beta \left[{\frak J}_{e_1} (t) {\frak J}_{e_1} (0)\right]
\end{equation*} 

A normal finite conductivity corresponds (in a Tauberian sense) to a positive finite limit of $\lambda^{2} {\frak L} (\lambda)$ as $\lambda \to 0$. In this case, the conductivity $\kappa^{1,1} (\{m_x\})$ is equal to
\begin{equation}
\label{eq:18}
\kappa^{1,1} (\{m_x\})=\gamma/d + \lim_{\lambda \to 0} \lim_{N\to \infty} \int_0^{\infty} dt  e^{-\lambda t} N^{-d} {\mathbb E}_\beta \left[ {\frak J}_{e_1}(t)\, {\frak J}_{e_1} (0) \right]dt 
\end{equation}

The right hand side of (\ref{eq:18}) is the sum of two terms. The first one is only due to the noise and is of no interest. The second one is the contribution of Hamiltonian dynamics to the conductivity and this is this term we investigate in the sequel.

All these computations are valid as soon as we can take the infinite volume limit $N \to \infty$ and then the limit $\lambda \to 0$. In the homogenous case ($m_x=m$ for all $x$), one can show that all the limits exist and one can compute explicitely $C(t)$ (see section \ref{sec:lubounds}). In the non homogenous case, we can only prove such a convergence up to subsequences (see corollary \ref{cor:1}). What we are able to do is to prove upper and lower bounds which indicate a finite and positive contribution of the Hamiltonian dynamics to the conductivity (see section \ref{sec:lubounds}).

\section{Homogenized infinite volume Green-Kubo formula}
\label{sec:GK}

In this section, we show that the homogenized infinite volume Green-Kubo formula for the thermal conductivity is well defined, positive and finite. We give also heuristic arguments showing that for almost all realization of masses $m=\{m_x\}$, the Green-Kubo formula (\ref{eq:18}) coincides with the homogenized infinite volume Green-Kubo formula $\kappa_{hom.}$ defined in (\ref{eq:GKinfty}).

 Assume that masses $m=\{m_x\}_{x \in {\mathbb Z}^d}$ are distributed according to an ergodic stationary probability measure ${\mathbb E}^*$. A typical configuration in the phase space is noted $\omega= (\bpi,\bq)= ((\pi_x)_{x \in \TT_N^d}, (q_x)_{x\in \TT_N^d})$. The masses of the finite volume dynamics are obtained from the infinite sequence $\{ m_x\}$ by the map identity $x \in \{0,\ldots, N-1\}^d \subset {\mathbb Z}^d \to \TT_N^d$. For any $z \in {\TT}_N^d$ and any function $f(\omega,m)$ the translation of $f$ by $z$ is defined by
\begin{equation*}
(\tau_z f) (\omega,m)= f(\tau_z \omega, \tau_z m), \quad \tau_z \omega= ((\pi_{z+x})_x, (q_{z+x})_x), \quad (\tau_z m)_x =m_{x+z}
\end{equation*}
Observe that the dynamics is invariant under the action of the group of translation $\tau_z$. 
We indicate the dependance of the instantaneous current on the masses and on the configuration by 
$$j_{0,e_1}^{a} (\omega,m)=\cfrac{1}{\sqrt{m_0}}\pi_0 \cdot (q_{e_1}-q_{-e_1})$$

By translation invariance of the dynamics we have
\begin{eqnarray*}
N^{-d}\int_0^t dt e^{-\lambda t} {\mathbb E}_\beta \left[ {\frak J}_{e_1}(t)\, {\frak J}_{e_1} (0) \right]\\
=N^{-d} \int_0^t dt e^{-\lambda t} \sum_{x,y \in \TT_N^d} {\EE}_\beta \left[ [\tau_x j^{a}_{0,e_1}](\omega_t, m) \, [\tau_y j^{a}_{0,e_1}](\omega_0,m)\right]\\
=\cfrac{1}{N^d} \sum_{x \in \TT_N^d} F_N (\tau_{x} m)
\end{eqnarray*}
where
\begin{equation*}
F_N (m)= \int_0^\infty dt e^{-\lambda t} \sum_{z \in \TT_N^d} \EE_\beta \left[ j^{a}_{0,e_1} (\omega_t , m) \,  [\tau_z j^{a}_{0,e_1}](\omega_0, m)\right]
\end{equation*}

In appendix, we explain how to define the dynamics starting from the infinite volume Gibbs measure $\mu_\beta$ and we show the dynamics is stationary w.r.t. $\mu_\beta$. The law of the dynamics is noted $\PP_\beta$. 

For fixed positive $\lambda$, we conjecture that as $N$ goes to infinity, the finite volume dynamics is closed to the infinite volume dynamics in the sense that
\begin{equation*}
\lim_{N \to \infty} \left| F_N (m)-  \int_0^\infty dt e^{-\lambda t} \sum_{z \in \ZZ^d} {\EE}_\beta \left[ j^{a}_{0,e_1} (\omega_t , m)\, [\tau_z j^{a}_{0,e_1}](\omega_0, m)\right] \right|=0
\end{equation*}

Then, by ergodic theorem, we have 
\begin{equation}
\label{eq:21}
\lim_{N \to \infty} N^{-d}\int_0^t dt e^{-\lambda t} {\mathbb E}_{\beta} \left[ {\frak J}_{e_1}(t),\, {\frak J}_{e_1} (0) \right] = {\EE}^* \left[  \int_0^\infty dt e^{-\lambda t} \sum_{z \in \ZZ^d} {\EE}_\beta \left[ j^{a}_{0,e_1} (t)\,\tau_z j^{a}_{0,e_1}(0)\right] \right]
\end{equation}

We are not able to prove the convergence (\ref{eq:21}) but we can prove the following existence theorem for the homogenized infinite volume Green-Kubo formula.

\begin{theo}
\label{th:AGK}
Assume that $\{m_x\}$ is stationary under ${\PP}^*$ and there are positive constants $\underline{m}$ and ${\overline m}$ such that
\begin{equation*}
{\mathbb P}^* ({\underline m} \leq m_x \leq {\overline m})=1
\end{equation*}
The Hamiltonian contribution to the homogenized Green-Kubo formula for the thermal conductivity $\kappa^{1,1}_{hom.} -\gamma/d$
\begin{equation}
\label{eq:GKinfty}
\kappa^{1,1}_{hom.} - \gamma/d =\lim_{\lambda \to 0} {\EE}^* \left[  \int_0^\infty dt e^{-\lambda t} \sum_{z \in \ZZ^d} {\EE}_{\beta} \left[ j^{a}_{0,e_1} (t) \tau_z j^{a}_{0,e_1}(0)\right] \right] 
\end{equation}
exists, is positive and finite.
 \end{theo}
 
 \begin{proof}
 The proof closely follows \cite{BLLO} and \cite{Ben}. With respect to the self consistent model of \cite{BLLO}, the symmetric part of the generator ${S}$ does not have a spectral gap. To overcome this difficulty, we prove in lemma \ref{lem:ua} that the antisymmetric part of the resolvent solution is an eigenfunction of ${S}$. It turns out that it is sufficient to conclude the proof.
 
 We define the following semi-inner product on $\LL^2 (\PP^* \otimes \mu_\beta)$
 \begin{eqnarray*}
 \ll f,g \gg &=& \sum_{z \in \ZZ^d} \left\{ \EE^* \left[ \mu_\beta( f \tau_z g) -\mu_\beta (f) \mu_\beta(g) \right]\right\}\\
 &=& \lim_{K \to \infty} \cfrac{1}{(2K+1)^d} \sum_{\substack{|x| \le K \\ |y| \le K}}\left\{ \EE^* \left[ \mu_\beta( \tau_x f \tau_y g) -\mu_\beta (f) \mu_\beta(g) \right]\right\}
 \end{eqnarray*}
We denote by ${\LL}^2_*$ the completion of the space of square integrable local functions w.r.t. this semi-inner product. The generator ${L}$ has the decomposition ${A} +\gamma {S}$ in antisymmetric and symmetric part in $\LL^2_*$. The $H_1$ norm corresponding to the symmetric part is denoted
\begin{equation*}
\| f \|_1^2 = \ll f, (-{S}) f \gg
\end{equation*} 
and ${\mathcal H}_1$ is the Hilbert space obtained by the completion of $\LL^2_*$ w.r.t. this norm.

Let $u_\lambda$ be the solution of the resolvent equation 
\begin{equation}
\label{eq:res}
\lambda u_\lambda -{L} u_\lambda = j^a_{0,e_1}
\end{equation} 
 
 We have to prove that $\ll u_\lambda, j_{0,e_1} \gg$ converges as $\lambda$ goes to $0$ and that the limit is positive and finite. 
 
 We multiply (\ref{eq:res}) by $u_\lambda$ and integrate w.r.t. $\ll\cdot,\cdot\gg$ and we get
 \begin{equation*}
 \lambda \ll u_\lambda, u_\lambda \gg +\gamma \| u_\lambda \|_1^2 = \ll u_\lambda, j^a_{0,e_1} \gg
 \end{equation*}
 Since ${S} (j^a_{0,e_1})= -j^a_{0,e_1}$ (see lemma \ref{lem:calcul}), by Schwarz inequality, we have
 \begin{equation*}
 \| u_\lambda \|_1^2 \leq C^2 \gamma^{-1}
 \end{equation*}
 and
 \begin{equation*}
 \lambda \ll u_\lambda, u_\lambda \gg \leq C^2 \gamma^{-1}
 \end{equation*}
Since $(u_\lambda)_{\lambda}$ is a bounded sequence in ${\mathcal H}_1$, we can extract a weakly converging subsequence in ${\mathcal H}_1$. We continue to denote this subsequence $(u_\lambda)_\lambda$ and we note $u_0$ the limit.

Let $u_\lambda (p,q)= u_\lambda^s (p,q) + u_{\lambda}^a (p,q)$ be the decomposition of $u_\lambda$ in its symmetric and antisymmetric part in the $p$'s. Since $j^a_{0,e_1}$ is antisymmetric in the $p$'s, we have that $\ll u_\lambda, j^a_{0,e_1} \gg = \ll u_\lambda^a , j^a_{0,e_1} \gg$. Furthermore ${S}$ preserves the parity in $p$ while it is inverted by ${A}$. We have the following decomposition
\begin{eqnarray*}
\lambda u_{\lambda}^s - \gamma {S} u_{\lambda}^s -{A} u_{\lambda}^a =0\\
\mu u_{\mu}^a - \gamma {S} u_{\mu}^a -{A} u_{\mu}^s =j^a_{0,e_1}
\end{eqnarray*} 

We multiply the first equality  by $u_\mu^s$ and the second by $u_\lambda^a$ and we use the antisymmetry of ${A}$. We get
\begin{equation*}
\ll u_\lambda^a, j^a_{0,e_1} \gg = \mu \ll u_{\mu}^a, u_\mu^a \gg + \lambda \ll u_{\lambda}^s, u_\lambda^s \gg + \gamma \ll u_\lambda, (-{S}) u_\mu \gg
\end{equation*}

In lemma \ref{lem:ua}, we prove that ${S} u_{\lambda}^a = -u_{\lambda}^a$. It follows that
\begin{equation*}
\ll u_\lambda^a , u_\lambda^a \gg = \| u_\lambda^a\|_1^2 = \|u_\lambda \|_1^2 - \|u_\lambda^s\|_1^2 \le C^2 \gamma^{-1}
\end{equation*}

Remark that $u_\lambda^a$ and $u_\lambda^s$ converge weakly in ${\mathcal H}_1$ respectively to $u_0^a$ and to $u_0^s$. We first take the limit as $\lambda \to 0$ and then as $\mu \to 0$ and we obtain
\begin{equation*}
\ll u_0, j^a_{0,e_1} \gg = \gamma \ll u_0, (- S) u_0 \gg
\end{equation*}
On the other hand, since ${S} j^a_{0,e_1} = -j^a_{0,e_1}$, we have
\begin{eqnarray*}
\ll u_0, j^a_{0,e_1} \gg = \lim_{\lambda \to 0} \ll u_\lambda , j^a_{0,e_1} \gg \nonumber \\
=\lim_{\lambda \to 0} \left[ \lambda \ll u_\lambda, u_\lambda \gg_*  + \ll u_\lambda, (-{A}) u_\lambda \gg + \gamma \ll u_\lambda, (- {S} ) u_\lambda \gg \right] \nonumber \\
=\lim_{\lambda \to 0} \left[ \lambda \ll u_\lambda, u_\lambda \gg_* + \gamma \ll u_\lambda, (- {S} ) u_\lambda \gg \right] \nonumber \\
\geq \lim_{\lambda \to 0} \lambda \ll u_\lambda, u_\lambda \gg + \gamma \ll u_0, (-S) u_0 \gg \nonumber
\end{eqnarray*}
where the last inequality follows from the weak convergence in ${ \mathcal H}_1$ of $(u_\lambda)_\lambda$ to $u_0$. It implies
\begin{equation*}
\lim_{\lambda \to 0} \lambda \ll u_\lambda, u_\lambda \gg =0
\end{equation*}
so that $u_\lambda$ converges strongly to $u_0$ in ${\mathcal H}_1$. Hence $\ll u_\lambda, j_{0,e_1} \gg$ converges to $\gamma \ll u_0, -S u_0 \gg$. Uniqueness of the limit follows by a standard argument.

The positivity and finiteness of the limit is postponed to lemma \ref{lem:amenophis}.
 \end{proof}
 
 \begin{lemma}
 \label{lem:ua}
 Let $u_\lambda$ be the solution of the resolvent equation
 \begin{equation*}
 \lambda u_\lambda -{L} u_\lambda = j^a_{0,e_1}
 \end{equation*}
Let $u_\lambda^a$ the antisymmetric part of $u_\lambda$ with respect to the $\pi$'s. $u_\lambda^a$ is such that
 \begin{equation*}
 {S} u_{\lambda}^a = -u_\lambda^a
 \end{equation*}
 \end{lemma}

\begin{proof}
Let $X$ be the closure in $\LL^2 (\mu_\beta)$ of the  space of polynomial functions in $\pi$ and $q$ of degree $2$. The generator ${L}$ transforms a polynomial function in a polynomial function and  conserves the degree so that the image of $X$ under ${L}$ is included in $X$. Since $j^a_{0,e_1}$ is in $X$, $u_\lambda$ is in $X$. Let $\varepsilon >0$ and consider $v_\varepsilon$ a polynomial function of degree $2$ such that
\begin{equation*}
\mu_\beta( [u_\lambda -v_\varepsilon ]^2) \le \varepsilon
\end{equation*} 
One easily shows that 
\begin{equation*}
\mu_\beta ([u_\lambda^a -v_\varepsilon^a]^2) \le \varepsilon
\end{equation*} 
Moreover, $v_\varepsilon^a$ is of the form
\begin{equation*}
v_\varepsilon^a = \sum_{x,y \in \ZZ^d} \rho(x,y) \pi_x q_y
\end{equation*}
where $\rho$ is a function with compact support. Since ${S} \pi_x = -\pi_x$, we have ${S} v_\varepsilon^a = -v_{\varepsilon}^a$. As $\varepsilon$ goes to $0$, we get
\begin{equation*}
S u_\lambda^a = -u_\lambda^a 
\end{equation*}
\end{proof}

\section{Lower and upper bounds for the Green-Kubo formula}
\label{sec:lubounds}

The canonical measure $\mu^N_\beta$ with temperature $T=\beta^{-1}$ and periodic boundary conditions on ${\mathbb T}_N^d$ is denoted by $<\cdot>$ and the scalar product associated in ${\mathbb L}^2(\mu^N_\beta)$ by $<\cdot,\cdot>$. 

The dynamics is given by (\ref{FP}) and $\{m_x\}_{x \in {\mathbb T}_N^d}$ is a sequence of positive masses bounded above and below by $\underline{m}$ and ${\overline m}$.  The total current in the first direction $e_1$ is given by
\begin{equation*}
{\frak J}_{e_1}=\omega \sum_{z} \cfrac{1}{\sqrt m_z} \pi_z \cdot (q_{z+e_1} -q_{z-e_1})
\end{equation*}

Before considering the non homogenous case, we compute briefly the time current-current correlations in the homogenous case (i.e. $m_x=m$ for all $x$). Let us define
\begin{equation*}
C_N (t)= (2 T^2 t N^d)^{-1}{\mathbb E}_\beta\left(\left[\int_0^t {\frak J}_{e_1} (s) ds\right]^2 \right)
\end{equation*} 
Since starting from $\mu_\beta$ the process is stationary we have
\begin{equation*}
C_N (t)= \cfrac{1}{T^2 t N^d}\int_0^t ds \int_{0}^s du {\mathbb E}_\beta \left[ {\frak J}_{e_1} (u), {\frak J}_{e_1} (0)\right]
\end{equation*} 
performing two integration by parts, one obtains that the Laplace transform ${\frak L}_N (\lambda)$ of $t C_N (t)$ is equal to  
\begin{equation*}
{\frak L}_N (\lambda)=\cfrac{1}{\lambda^2 N^d T^2} \int_0^\infty dt e^{-\lambda t} {\mathbb E}_\beta \left[{\frak J}_{e_1} (t), {\frak J}_{e_1} (0)\right]
\end{equation*} 
This last quantity is equal to
\begin{equation*}
\cfrac{1}{\lambda^2 T^2 N^d} < {\frak J}_{e_1}, (\lambda -L)^{-1} {\frak J}_{e_1}>
\end{equation*}
A simple but crucial computation shows that
\begin{equation*}
(\lambda -L)^{-1} {\frak J}_{e_1}= \cfrac{{\frak J}_{e_1}}{\lambda +\gamma}
\end{equation*} 
so that ${\frak L}_N (\lambda)$ is given by
\begin{equation*}
{\frak L}_N (\lambda)= \cfrac{<{\frak J}_{e_1}, {\frak J}_{e_1}>}{T^2 N^d \lambda^2(\lambda + \gamma)}
\end{equation*}
Let $D_N=D_N (\omega,\nu)$ be the constant
\begin{equation}
\label{eq:D_N}
D_N= T^{-1} \omega^2 \sum_{k=1}^d <(q_{e_1}^k -q_{-e_1}^k)^2>=\cfrac{1}{N^d} \sum_{\xi \in \TT_N^d} \left( \cfrac{4\omega^2 \sum_{j=1}^d \sin^2 (\pi \xi^j /N)}{\nu +4\omega \sum_{j=1}^d \sin^2 (\pi \xi^j /N)}\right)
\end{equation}
One computes easily $<{\frak J}_{e_1}, {\frak J}_{e_1}>$ and after inversion of the Laplace transform, one gets
\begin{equation*}
C_N (t)=\cfrac{D_N}{\gamma m} \left( 1+ \cfrac{1}{\gamma t} (1-e^{-\gamma t}) \right) 
\end{equation*}
As $N$ and then $t$ goes to infinity, it converges to the constant $D/(\gamma m)$ where
\begin{equation}
\label{eq:D}
D = \int_{\xi \in [0,1]^d} \left( \cfrac{4\omega^2 \sum_{j=1}^d \sin^2 (\pi \xi^j)}{\nu +4\omega \sum_{j=1}^d \sin^2 (\pi \xi^j)}\right)d\xi^1 \ldots d\xi^d
\end{equation}

One concludes that the thermal conductivity is given by
\begin{equation*}
\kappa_{1,1} (\{m\})= \cfrac{D}{\gamma m} +\cfrac{\gamma}{d}
\end{equation*}

Observe that if the noise becomes weaker (i.e. $\gamma \to 0$), we obtain a purely homogenous harmonic chain and the thermal conductivity is infinite. \\

In the non homogeneous case, we are not able to obtain explicitly the solution $h$ of the resolvent equation $(\lambda -L)h={\frak J}_{e_1}$ but we obtain upper and lower bounds for the Laplace transform of the time current-current correlations function which indicate a finite positive Hamiltonian contribution to the conductivity (for any bounded below and above sequence of masses). This is the content of the following theorem.

\begin{theo}
\label{th:QGK}
There exists a positive constant $C>0$ independent of $\lambda$ and $N$ such that
\begin{equation}
\begin{split}
C^{-1} \le \liminf_{\lambda \to 0} \liminf_{N \to \infty} \int_0^{\infty} e^{-\lambda t} N^{-d} {\mathbb E}_\beta{\left[{\frak J}_{e_1}(t),\, {\frak J}_{e_1} (0) \right]}dt \\
 \le \limsup_{\lambda \to 0} \limsup_{N\to \infty} \int_0^{\infty} e^{-\lambda t} N^{-d} {\mathbb E}_\beta {\left[{\frak J}_{e_1}(t),\, {\frak J}_{e_1} (0) \right]}dt \le C
\end{split}
\end{equation}
\end{theo}

\begin{proof}
We have
\begin{equation}
\label{eq:lap}
\int_0^{\infty} e^{-\lambda t} N^{-d} {\mathbb E}{\left[{\frak J}_{e_1} (t),\, {\frak J}_{e_1} (0) \right]}dt = \cfrac{1}{N^d} < {\frak J}_{e_1}, (\lambda -L)^{-1} {\frak J}_{e_1}>
\end{equation}
The proof is based on a variational formula for the right hand side of (\ref{eq:lap}) and a suitable choice of test functions over which the supremum is carried. We need to introduce Sobolev norms associated to the operator ${\gamma S}$. $H_{1,\lambda}$ norm is defined by
\begin{equation*}
\|f\|_{1,\lambda}^2 = <(\lambda -\gamma S)f,f>
\end{equation*}
and the $H_{-1,\lambda}$ norm is the dual norm of the $H_{1,\lambda}$ norm in ${\mathbb L}^2 (\mu_\beta^N)$
\begin{equation*}
\|f\|_{-1,\lambda}^2 = <(\lambda -\gamma S)^{-1}f,f>=\sup_{g} \left\{ 2<f,g>- <g, (\lambda -\gamma S) g> \right\}
\end{equation*}
where the supremum is carried over local smooth functions $g(\bpi,\bq)$ from $({\mathbb R}^d \times {\mathbb R}^d)^{{\mathbb T}_N^d}$ into $\mathbb R$.

Recall now that the generator $L$ is given by the sum $A+\gamma S$ where $A$ is antisymmetric and $S$ is symmetric (in ${\mathbb L}^2 (\mu^N_\beta)$). The variational formula is the following
\begin{equation}
\label{eq:lap2}
\left<  {\frak J}_{e_1}, (\lambda -L)^{-1} {\frak J}_{e_1} \right>=\sup_{u} \left\{ 2<u,{\frak J}> -\|u \|_{1,\lambda}^2 -\|Au\|_{-1,\lambda}^2\right\}
\end{equation}
where the supremum is carried over the set of smooth functions $u(\bpi,\bq)$ from $({\mathbb R}^d \times {\mathbb R}^d)^{{\mathbb T}_N^d}$.

\vspace{0,5 cm} 
  
{\textbf {Upper bound :}}
By neglecting the term $\|Au\|_{-1,\lambda}^2$ in the variational formula (\ref{eq:lap2}), we get
\begin{equation*}
<{\frak J}_{e_1}, (\lambda -L)^{-1} {\frak J}_{e_1}> \leq \| {\frak J}_{e_1} \|^2_{-1, \lambda}=(\lambda+\gamma)^{-1} <{\frak J}_{e_1}, {\frak J}_{e_1}>
\end{equation*}
The last equality follows from the fact that $S{\frak J}_{e_1}=-{\frak J}_{e_1}$ (see lemma \ref{lem:calcul}). Since $<\cdot>$ is the Gaussian measure $\mu_\beta$, the $\pi$'s are Gaussian product independent variables and we have
\begin{eqnarray*}
 <{\frak J}_{e_1}, {\frak J}_{e_1}>&=& \omega^2 \sum_{x,y}\sum_{k,\ell} \cfrac{1}{\sqrt{m_x m_y}} \left< {\pi}^k_x {\pi}^\ell_y (q^k_{x+e_1} -q^k_{x-e_1})(q^\ell_{y+e_1} -q^\ell_{y-e_1}) \right>\\
 &=& \omega^2 \sum_{x}\sum_{k} \cfrac{1}{m_x} \left< ({\pi}^k_x)^2 (q^k_{x+e_1} -q^k_{x-e_1})^2\right>\\
&\le& N^d D_N (\omega,\nu) T^2 {\underline m}^{-1}
\end{eqnarray*}
where $D_N$ is defined by (\ref{eq:D_N}). 

\vspace{0,5cm}

{\textbf {Lower bound :}}

For the lower bound we use again the following variational formula 
\begin{equation*}
\left<  {\frak J}_{e_1}, (\lambda -L)^{-1} {\frak J}_{e_1} \right>=\sup_{u} \left\{ 2<u,{\frak J}_{e_1}> -\|u \|_{1,\lambda}^2 -\|Au\|_{-1,\lambda}^2\right\}
\end{equation*}
and we take the test function $u$ in the form
\begin{equation*}
u=\rho \sum_{x \in \TT_N^d} \sqrt{m_x} {\pi}_x \cdot (q_{x+e_1} - q_{x-e_1})
\end{equation*}
with $\rho$ a positive constant. A simple computation shows that
\begin{equation*}
Au=\rho \sum_{x \in \TT_N^d} \left( \sqrt{\cfrac{m_{x+e_1}}{m_{x}}}-\sqrt{\cfrac{m_x}{m_{x+e_1}}}\right)\pi_x \cdot \pi_{x+e_1}
\end{equation*}
Moreover, by lemma \ref{lem:calcul}, one has
\begin{equation*}
Au=\rho \left(\lambda +\gamma (2+d^{-2})\right)^{-1} (\lambda-\gamma S)\left(  \sum_{x \in \TT_N^d} \left( \sqrt{\cfrac{m_{x+e_1}}{m_{x}}}-\sqrt{\cfrac{m_x}{m_{x+e_1}}}\right)\pi_x \cdot \pi_{x+e_1}
\right)
\end{equation*}
Hence, the $H_{-1,\lambda}$ norm of $Au$ is easy to compute and given by
\begin{equation*}
\|Au\|_{-1,\lambda}^2 = \cfrac{\left< Au, \, Au\right>}{\gamma(2+d^{-2})+\lambda}
\end{equation*}
This last quantity is equal to
\begin{equation*}
<Au,Au>= \rho^2 \sum_{x \in \TT_N^d} \sum_{k=1}^d \left( \sqrt{\cfrac{m_x}{m_{x+e_1}}} -\sqrt{\cfrac{m_{x+e_1}}{m_x}}\right)^2\left< (\pi_x^k)^2 (\pi_{x+e_1}^k)^2 \right> 
\end{equation*}

By lemma \ref{lem:calcul}, one has
\begin{eqnarray*}
\| u\|_{1,\lambda}^2&=&<u,(\lambda -\gamma S)u>= (\gamma+\lambda)<u^2>\\
&=& (\gamma+\lambda) \rho^2 \sum_{x \in \TT_N^d} \sum_{k=1}^d m_x <(\pi_x^k)^2 (q_{x+e_1}^k-q_{x-e_1}^k)^2> 
\end{eqnarray*}
and 
\begin{equation*}
<u,{\frak J}_{e_1}>= \rho \sum_{x \in \TT_N^d} \sum_{k=1}^d <(\pi_x^k)^2 (q^k_{x+e_1} -q^k_{x-e_1})^2>
\end{equation*}
Hence, we get
\begin{equation}
\label{eq:vf999}
\begin{array}{l}
\cfrac{1}{N^d} \left<{\frak J}_{e_1}, (\lambda -L)^{-1} {\frak J}_{e_1} \right>\\
\geq \rho C_0 -C_1 \rho^2\\
\end{array}
\end{equation}
with $C_0, C_1$ positive constants given by
\begin{equation*}
C_0 = 2 D_N, \quad C_1= (\gamma+\lambda) {\bar m} D_N +\cfrac{dT^2 \left( \cfrac{\underline m}{\bar m} -\cfrac{\bar m}{\underline m}\right)^2}{\gamma (2+d^{-2})+\lambda}
\end{equation*}
With the optimal choice $\rho=C_0/2C_1$ we get
\begin{equation*}
 \liminf_{\lambda \to 0} \liminf_{N \to \infty} \int_0^{\infty} e^{-\lambda t} N^{-d} {\mathbb E}_\beta{\left[{\frak J}_{e_1}(t),\, {\frak J}_{e_1} (0) \right]}dt  \ge D^{-2}\left\{\gamma {\bar m} D +\cfrac{dT^2 \left( \cfrac{\underline m}{\bar m} -\cfrac{\bar m}{\underline m}\right)^2}{\gamma (2+d^{-2})}\right\}^{-1}
\end{equation*}
with $D$ given by (\ref{eq:D}).
\end{proof}

\begin{remark}
1. As $\gamma \to 0$, the lower bound obtained is of the form $C_0 \gamma$ and the upper bound of the form $C_1 \gamma^{-1}$. Moreover the large pinning limit $\nu \to \infty$ gives null upper and lower bounds.\\
2. The same result can be proved for the microcanonical version of the Green-Kubo formula meaning with the replacement of the canonical measure $\mu^N_\beta= Z_\beta^{-1} \exp(-\beta H)$ by the microcanonical measure which is nothing else than the uniform measure on the shell of constant energy $\left\{ H=N^d \beta^{-1}\right\}$.\\
3. The upper bound is in fact valid for a general disordered {\textit{anharmonic}} chain with interaction potentials $V$ and pinning potential $W$. The reason is that we have still $S(\frak{J}_{e_1})=-{\frak J}_{e_1}$ so that the upper bound remains in force as soon as $N^{-d} <{\frak J}_{e_1}, {\frak J}_{e_1}> \le C$, which is equivalent to $<[V' (q_{e_1}-q_0)]^2> \le C$ uniformly in $N$. 
\end{remark}

Here we deduce from the upper bound obtained for the Laplace transform of the time current-current correlations function an upper bound for the function itself. 

\begin{cor}
\label{cor:1}
We have
\begin{equation*}
(2 t N^d)^{-1}{\mathbb E}_{\beta} \left(\left[\int_0^t {\frak J}_{e_1} (s) ds\right]^2 \right) \leq \cfrac{C}{\gamma + t^{-1}}
\end{equation*}
where $C$ is a positive constant depending on the parameters of the system.
\end{cor}

\begin{proof}
It is a simple consequence of a general argument valid for Markov processes (\cite{LY}, lemma 6.1).
\end{proof}

\begin{lemma}
\label{lem:amenophis}
The Hamiltonian contribution to the infinite volume homogenized Green-Kubo formula $\kappa^{1,1}_{hom.} -\gamma/d$ (see (\ref{eq:GKinfty})) is positive and finite.
\end{lemma}

\begin{proof}
The proof is a simple rephrasing of the proof of the theorem above. We have just to replace the inner product $<\cdot,\cdot>$ by the inner product with translations $\ll \cdot, \cdot \gg$. 

We have
\begin{equation*}
{\EE}^* \left[  \int_0^\infty dt e^{-\lambda t} \sum_{z \in \ZZ^d} \EE_\beta \left[ j^{m}_{0,e_1} (t) j^{m}_{z,z+e_1}(0)\right] \right]= \ll j_{0,e_1}^a, (\lambda -L)^{-1} j_{0,e_1}^a \gg
\end{equation*}
and again a variational formula for the last term is available
\begin{equation*}
\ll j_{0,e_1}^a, (\lambda -L)^{-1} j_{0,e_1}^a \gg =\sup_{u} \left\{ 2\ll u, j_{0,e_1}^a \gg -\|u \|_{1,\lambda}^2 -\|Au\|_{-1,\lambda}^2\right\}
\end{equation*}  
where the supremum is now carried over local smooth function $u(\bpi,\bq)$ and $H_{1,\lambda}$ and $H_{-1,\lambda}$ norms are defined by
\begin{equation*}
\| f\|_{\pm1, \lambda}^2 =\ll f, (\lambda -\gamma S)^{\pm 1} f \gg
\end{equation*}

To obtain the upper bound, we neglect the term coming from the antisymmetric part $Au$ and remark that $S j_{0,e_1}^a = -j_{0,e_1}^a$.  For the lower bound, we use the same test function $u$ as in the theorem above:
$$u=\rho \sum_{x \in {\ZZ}^d} {\sqrt m_x} \pi_x \cdot (q_{x+e_1} -q_{x-e_1})$$
\end{proof}

\begin{remark}
Suppose $\{m_x\}$ forms a stationary sequence of random masses with law ${\mathbb E}^*$ and let us denote $\mu={\mathbb E}^{*} (1/m_0)$. Assume that 
\begin{equation*}
\cfrac{\gamma}{d} + \lim_{\lambda \to 0} \lim_{N \to \infty} \int_0^{\infty} e^{-\lambda t} N^{-d} {\mathbb E}{\left[{\frak J}(t),\, {\frak J}(0) \right]}dt =\kappa^{1,1} (\{m_x\})
\end{equation*}  
exists. We expect  $\kappa^{1,1} (\{m_x\})$ to depend only on the statistics of $m_x$ and not on the particular realization of random masses and to be equal to the infinite volume homogenized Green-Kubo formula $\kappa^{1,1}_{hom.}$ (see (\ref{eq:21})). Upper bounds show in fact that
\begin{equation*}
\kappa^{1,1}_{hom.} \leq \kappa^{1,1} (\{1/\mu\})
\end{equation*}
where we recall that $\kappa^{1,1} (\{1/\mu\})$ is the thermal conductivity of the homogenous chain with mass $1/\mu$ (also called in the homogenization literature the "effective conductivity"). It is an open problem to know if this inequality is in fact an equality or not.
\end{remark}

\section{Appendix}

\begin{lemma}
\label{lem:calcul}
Let $x \in \TT_N^d$ and $k,\ell,m \in \{1, \ldots,d\}$. We have
\begin{itemize}
\item
$S(\pi_x^\ell)=- \pi_x^\ell$
\item
$S(|\pi_x|^2)=d^{-1} \Delta (|\pi_x^2|)$
\item
$S(\pi_x^k \pi_{x+e_m}^\ell)= -\left( 2+d^{-2}\right)\pi_{x}^k \pi_{x+e_m}^\ell$
\end{itemize}
\end{lemma}

\begin{proof}
It is a simple computation.
\end{proof}

We prove here the existence of a measurable set $\Omega_0$ of initial conditions with full measure w.r.t. to the infinite volume Gibbs measure $\mu_\beta$ such that the infinite volume dynamics starting from $\omega \in \Omega_0$ exist. This defines a strongly continuous semigroup $(P_t)_{t \geq 0}$ on ${\mathbb L}^2 (\Omega_0, \mu_\beta)$ with generator $L$. Moreover the set of square integrable local smooth functions ${\mathcal D}$ is a core for $L$.  These arguments are by now standard (see \cite{FS}) and we repeat them for the convenience of the reader.

The dynamics is given by the following stochastic differential equations:
\begin{equation}
\label{eq:sde}
\begin{cases}
dq_x ={m_x}^{-1/2} \pi_x dt\\
d\pi_x = m_{x}^{-1/2} (\omega \Delta q_x -\nu q_x )dt  -\gamma {\pi}_x dt +\sqrt{\gamma} {\pi}_{x+1} dW_{x}-{\sqrt \gamma} \pi_{x-1} dW_{x-1}
\end{cases}
\end{equation}
where $\{ W_x \, ; \, x \in {\mathbb Z} \}$ are independent standard Brownian motions. We note ${\mathcal F}_t$ the $\sigma$-algebra generated by $\{ W_{x} (s), s \leq t \, ; x \in {\mathbb Z} \}$.

The first problem is do define the infinite volume Gibbs measure $\mu_\beta$. Indeed, it is well known that 
\begin{equation*}
\mu_\beta^N (q_0^2) = 
\begin{cases}
{\mathcal O} (1) \text{ if $\nu >0$ or $d\geq 3$}\\
{\mathcal O} (\log N) \text{ if $\nu=0$ and $d=2$}\\
{\mathcal O} (N) \text{ if $\nu=0$ and $d=1$}
\end{cases}
\end{equation*}
Hence, in dimension $d=1,2$, if $\nu=0$, the infinite volume Gibbs measure is not well defined. To overcome this difficulty we go over the gradient field $\eta_{(x,y)}=q_x - q_y, |x-y|=1$, which has zero (discrete) curl. Let $\chi$ the set of vector fields $\eta$ on $\ZZ^d$ with zero curl. An infinite volume Gibbs measure  $\mu$ on $\chi$ is defined by the conditions $\mu(\eta_{(x,y)}^2) < +\infty$ and via DLR equations. One can prove the following lemma
\begin{lemma}[\cite{FS}, theorems 3.1 and 3.2]
There exists a unique shift ergodic Gibbs measure $\mu_\beta$ on $\chi$ such that
\begin{equation}
\int_{\chi} \eta_{(0,e_j)} d\mu_\beta (\eta) =0
\end{equation} 
\end{lemma}

Clearly the dynamics for $(\bpi, \bq)$ in (\ref{eq:sde}) can be read as a dynamics for the gradient field $\eta_{(x,y)}$. Moreover the quantities of interest like the current are functions of the $\eta$'s. Hence, only the existence of the dynamics for the gradient field is needed. Nevertheless to simplify the argument we restrict the proof to the one dimensional pinned case for which $\mu_\beta (q_x^2) < +\infty$ for any $x \in \ZZ$.

Let $\Omega=({\mathbb R} \times {\mathbb R})^{\mathbb Z}$ be the configuration space equipped with the product topology. A typical configuration is of the form $\omega=(\pi_x, q_x)_{x \in \mathbb Z}$ with $\pi_x= m_x^{-1/2} p_x$ and $q_x , p_x$ the position and momentum of the atom $x$ with mass $m_x$ we assume to be uniformly bounded above an below by finite positive constants . 

\begin{lemma}[Existence of the infinite volume dynamics]
There exists a measurable set $\Omega_0 \subset ({\mathbb R} \times {\mathbb R})^{\mathbb Z}$ with full measure w.r.t. $\mu_\beta$ such that for any initial condition $\omega (0) \in \Omega_0$ there exists a ${\mathcal F}_t$-adapted continuous stochastic process $\{ \omega (t)\}$ which satisfies (\ref{eq:sde}). Moreover, $\mu_\beta$ is a stationary probability measure for $\{ \omega (t)\}$.
\end{lemma}

\begin{proof}
We introduce $\Omega_0 =\left\{ \omega \in \Omega; \| \omega \|^2=\sum_{k \in {\mathbb Z}} e^{-|k|}(|\pi_k|^2 + |q_k|^2) < +\infty \right\}$ which is a measurable set with full measure w.r.t. $\mu_\beta$. Since the right hand side of (\ref{eq:sde}) is uniformly Lipschitz continuous w.r.t. the $\ell^2$ norm $\| \cdot \|$ an iteration procedure gives the existence and uniqueness of a strong solution to (\ref{eq:sde}). The fact that $\mu_\beta$ is invariant for $\{\omega (t)\}$ is standard (\cite{FS}).
\end{proof}

By this way we define a semigroup $(P_t)_{t \geq 0}$ on the Banach space ${\mathcal B} (\Omega_0)$ of bounded measurable functions on $\Omega_0$. For any $f \in {\mathcal B} (\Omega_0)$, we have
\begin{equation*}
\forall \omega \in \Omega_0, \quad (P_t f) (\omega) = {\mathbb E}_{\omega} \left[ f(\omega_t)\right]
\end{equation*}
where $\omega_t$ is the strong solution of (\ref{eq:sde}). Moreover $(P_t)_t$ is contractive w.r.t. the ${\mathbb L}^2$-norm associated to $\mu_\beta$. It follows that $P_t$ can be extended to a semi-group of contraction on ${\mathbb L}^2 (\Omega, \mu_\beta)$. 

Let ${\mathcal D}_0$ be the set of local smooth bounded functions on $\Omega$. By continuity of the paths $\omega_t$ and the bounded convergence theorem, we have that for any $\phi \in {\mathcal D}_0$, 
\begin{equation*}
\lim_{t \to 0} \mu_\beta ((P_t \phi - \phi )^2)=0
\end{equation*}
Since any function in ${\mathbb L}^2 (\Omega, \mu_\beta)$ can be approximated by a sequence of elements of ${\mathcal D}_0$ and $P_t$ is contractive, it follows that $P_t$ is a strongly continuous semigroup of contractions on ${\mathbb L}^2 (\Omega, \mu_\beta)$.  

Let $\mathcal D$ be the space of smooth (not necessarily  bounded) square integrable local functions on $\Omega$. By It\^o's formula, we have that for any $\phi \in {\mathcal D}$
\begin{equation*}
\forall \omega \in \Omega_0, \quad (P_t \phi) (\omega) = (P_0 \phi) (\omega) + \int_0^t (P_s L \phi) (\omega) ds
\end{equation*} 
where $L$ is the formal generator defined in section $2$. This shows that any $\phi \in {\mathcal D}$ belongs to the domain of the generator ${\hat L}$ of the ${\mathbb L}^2$-semigroup $(P_t)_{t \geq 0}$ and that $L$ and $\hat L$ coincide on ${\mathcal D}$. By lemma 2.11 and proposition 3.1 of \cite{EK},  ${\mathcal D}$ is a core for ${\hat L}$. we have proved the following lemma

\begin{lemma}
\label{lem:core}
There exists a closed extension of $L$ in ${\mathbb L}^2 (\Omega, \mu_\beta)$ such that the space ${\mathcal D}$ of square integrable smooth local functions on $\Omega$ is a core. This closed extension is the generator of the strongly continuous semigroup $(P_t)_t$ defined above. 
\end{lemma}

\end{document}